\documentclass[reprint, 11pt]{elsarticle}

\usepackage{fullpage}
\usepackage{amsfonts}
\usepackage{latexsym,amsmath, amssymb}
\usepackage{amsthm}

\usepackage[mathscr]{eucal}
\usepackage{color}

\theoremstyle{plain}
\newtheorem{theorem}{Theorem}[section]
\newtheorem{corollary}{Corollary}[section]
\newtheorem{lemma}{Lemma}[section]
\newtheorem{proposition}{Proposition}[section]

\newcommand{\R}{\mathbb{R}}

\newcommand{\pz}{\sqrt{e^{2\phi}+|p|^2}}

\newcommand{\mcL}{\mathcal{L}}

\journal{Journal Differential Equations}

\begin{document}

\begin{frontmatter}

\title{Spatially homogeneous solutions of the \\ Vlasov-Nordstr\"om-Fokker-Planck system}


\author{Jos\'e Antonio Alc\'antara Felix}
\address{Department of Applied Mathematics\\ University of Granada\\ Granada, Spain\\ jaaf@correo.ugr.es}

\author{Simone Calogero}
\address{Department of Mathematics\\ Chalmers Institute of Technology\\ University of Gothenburg\\ Gothenburg, Sweden\\ calogero@chalmers.se}

\author{Stephen Pankavich}
\address{Department of Applied Mathematics and Statistics\\ Colorado School of Mines\\ Golden, CO USA\\ pankavic@mines.edu}

\begin{abstract}
The Vlasov-Nordstr\"{o}m-Fokker-Planck system describes the evolution of self-gravitating matter experiencing collisions with a fixed background of particles in the framework of a relativistic scalar theory of gravitation. 
We study the spatially-homogeneous system and prove global existence and uniqueness of solutions for the corresponding initial value problem in three momentum dimensions.  Additionally, we study the long time asymptotic behavior of the system and prove that even in the absence of friction, solutions possess a non-trivial asymptotic  profile. An exact formula for the long time limit of the particle density is derived in the ultra-relativistic case.
\end{abstract}


\end{frontmatter}

\section{Introduction}
The Vlasov-Nordstr\"om-Fokker-Planck (VNFP) system has been introduced in~\cite{AC} as a simplified model for the relativistic diffusion dynamics of self-gravitating particle systems.
In the absence of friction, the VNFP system is given by
\begin{align}
&\partial_t f+\frac{p}{\pz}\cdot\nabla_x f-\frac{e^{2\phi}\nabla_x\phi}{\pz}\cdot\nabla_pf=\sigma\, e^{2\phi}\partial_{p^{i}}\left(\frac{e^{2\phi}\delta^{i j}+p^{i}p^{j}}{\sqrt{e^{2\phi}+|p|^2}}\partial_{p^{j}}f\right),\label{FP}\\
&\partial_{t}^2\phi-\Delta_x\phi=-e^{2\phi}\int_{\mathbb{R}^3}\!\frac{f}{\sqrt{e^{2\phi}+|p|^2}}\,dp,\quad t>0,\ x\in\R^3,\ p\in\R^3,\label{nordstrom}
\end{align}
where $f=f(t,x,p)$ is the particle density in phase space, $\phi=\phi(t,x)$ is the Nordstr\"om gravitational potential generated by the particles, and $\sigma>0$ is the diffusion constant. The remaining physical constants, i.e., the speed of light $c$, the mass $m$ of the particles, and the gravitational constant $G$, have been set equal to one. The physical interpretation of a solution of~\eqref{FP}-\eqref{nordstrom} is as follows. Spacetime is curved by the action of the gravitational forces and is given by the manifold $(\R^4,g)$, where $g$ is the conformally Minkowskian metric $g=\exp(2\phi)\eta$. 
In the collisionless case (i.e., for $\sigma=0$), the VNFP system reduces to the Nordstr\"om-Vlasov system~\cite{CR,C1}, a toy model for the full general relativistic Einstein-Vlasov system~\cite{hakan}. In contrast to the collisionless case, particles undergoing diffusion no longer move along the geodesics of spacetime. Their trajectories are defined through the system of stochastic differential equations naturally associated to the Fokker-Planck equation~\eqref{FP} via Ito's formula.

A consistent theory for the diffusion dynamics of particle systems in General Relativity has been proposed in~\cite{C2}, but due to the well-known complexity of the Einstein field equations, it seems wise to deal first with the analysis of the system~\eqref{FP}-\eqref{nordstrom}. The VNFP system already captures some of the essential features of relativistic gravitational systems undergoing diffusion: the hyperbolic character of the field equation,  the invariance under Lorentz transformations, and the space-time dependence of the diffusion matrix. These features distinguish the model under study from the Vlasov-Poisson-Fokker-Planck system, which is the non-relativistic analogue of the VNFP system~\cite{bou, bdol, csv, dol, dr2, gsz,ono}. While the non-relativistic problem has been investigated for a long time, the interest on relativistic diffusion models has only recently started to increase~\cite{AC, C2, DH1, DH2, frajan, herr, herr2, PM, PM2}.

In this paper we make a further simplification by restricting the discussion to spatially homogeneous solutions $(f=f(t,p), \phi=\phi(t))$, for which the VNFP system becomes (setting $\sigma=1$)
\begin{align}
&\partial_t f=\, e^{2\phi}\partial_{p^{i}}\left(\frac{e^{2\phi}\delta^{i j}+p^{i}p^{j}}{\sqrt{e^{2\phi}+|p|^2}}\partial_{p^{j}}f\right),\label{FPhom}\\
&\ddot{\phi}=-e^{2\phi}\int_{\mathbb{R}^3}\!\frac{f}{\sqrt{e^{2\phi}+|p|^2}}\,dp,\quad t>0,\ p\in\R^3.\label{nordstromhom}
\end{align}
Our results concern the global existence and uniqueness of solutions of the system~\eqref{FPhom}-\eqref{nordstromhom} and their asymptotic behavior as $t \to \infty$. Remarkably, and in contrast to the non-relativistic case~\cite{csv}, we find that the particle density $f$ does not vanish in the limit $t\to \infty$, as one would expect from the absence of a friction term in the diffusion equation. 
Moreover we show the gravitational potential $\phi$ blows-up to $-\infty$ as $t\to\infty$. 
Of course, the asymptotic behavior of the potential is the key to understanding the non-trivial large time limit of the particle density, as it implies that the action of the diffusion operator in the right  side of~\eqref{FPhom} becomes weaker and weaker as $t\to \infty$. To understand this mechanism in a simpler context, consider the non-autonomous heat equation
\begin{equation}\label{example}
\partial_tu=\lambda(t)\Delta_x u,\quad t>0,\ x\in\R^3,
\end{equation}
where $\lambda(t)$ is a smooth positive function integrable on $(0,\infty)$.
Upon introducing the change of variables $\tau(t)=\int_0^t\lambda(s)\,ds$, equation~\eqref{example} transforms into the standard, autonomous heat equation. It follows that the solution of~\eqref{example} with initial datum $u(0,x)=u_\mathrm{in}(x)$ is given by
\[
u(t,x)=\frac{1}{(4\pi\tau(t))^{3/2}}\int_{\R^3}u_0(y)\,e^{-\frac{|x-y|^2}{4\tau(t)}}\,dy.
\]
Hence, as $t\to\infty$,
\[
u(t,x)\sim\frac{1}{(4\pi\tau_\infty)^{3/2}}\int_{\R^3}u_0(y)\,e^{-\frac{|x-y|^2}{4\tau_\infty}}\,dy,\quad\text{where}\ \tau_\infty=\lim_{t\to\infty}\tau(t)=\int_0^\infty \lambda(s)\,ds<\infty,
\]
i.e., the solution has a non-trivial asymptotic profile.

The paper proceeds as follows.  In Section~\ref{mainresec} we state and prove our main result.  Namely, we establish the global existence and uniqueness of weak solutions to (\ref{FPhom})-(\ref{nordstromhom}), derive the asymptotic behavior of the scalar field (in particular $\phi\to -\infty$, as $t\to \infty$, linearly in time) and show that the particle density does not vanish as $t\to\infty$.
Since the differential operator in the right side of~\eqref{FPhom} is not uniformly elliptic and has time dependent coefficients, the standard theory for parabolic equations does not apply in our case, and we shall need to rely on stochastic methods to prove existence of solutions.   Section $3$ is then devoted to a detailed study of the asymptotic profile of the particle density in the ultra-relativistic regime, i.e., when the particle density solves the equation
\begin{equation}\label{ultrarelFPhom}
\partial_t f=\, e^{2\phi}\partial_{p^{i}}\left(\frac{p^{i}p^{j}}{|p|}\partial_{p^{j}}f\right),
\end{equation}
instead of~\eqref{FPhom}. Our main result is Theorem~\ref{hasym}, where we show that solutions of~\eqref{ultrarelFPhom} satisfy  $f\to f_\infty$ in $L^\infty$ as $t\to\infty$, where $f_\infty(p)>0$ is given by the solution of the linear ultra-relativistic Fokker-Planck equation evaluated at the finite time $\tau_\infty=\|e^{2\phi}\|_{L^1(\R^3)}$.
In particular, we are able to compute this limit $f_\infty$ explicitly.
It remains an interesting open problem to prove the analogous long time behavior of solutions in the purely relativistic case (\ref{FPhom}).
Finally, details of more technical computations are included within an appendix.

\section{Main result}~\label{mainresec}
We begin by fixing some notation. Given $T>0$ we denote by $C_T$ any positive constant that depends on the time interval $[0,T]$. If $\sup_{T>0}C_T<\infty$, we denote $C_T\equiv C$.  The diffusion matrix in the Fokker-Planck equation~\eqref{FPhom} will be denoted by $D[\phi]$ with entries
\begin{equation}\label{defD}
D^{ij}[\phi]=\frac{e^{2\phi}\delta^{i j}+p^{i}p^{j}}{\sqrt{e^{2\phi}+|p|^2}}.
\end{equation} 
In the appendix we collect some identities and estimates satisfied by $D[\phi]$ that are used throughout the paper. We use the index summation rule, which means that an index appearing twice in an expression, once in the lower and once in the upper position, is summed over $\{1,2,3\}$, e.g., 
\[
A^{ij} B_{jk}=\sum_{j,k=1}^3 A^{ij}B_{jk}.
\]
Moreover, indices are raised and lowered with the Kronecker symbol, e.g., $D^i_{\ j}=D^{ik}\delta_{kj}$, $p_i=\delta_{ij}p^j$.
The Banach space $X$ is defined throughout as
\[
X=\{g:\R^3\to \R: g\in L^1\cap L^2,\ \nabla g\in L^2,\ \ \mathrm{and} \ \ p\to |p| g(p)\in L^1\},
\]
with the usual notation for Lebesgue and Sobolev spaces.
The ultimate purpose of this section is to prove the following theorem.

\begin{theorem}\label{global}
Given  $\phi_\mathrm{in},\psi_\mathrm{in}\in\R$ and $f_\mathrm{in}\in X$ such that $f_\mathrm{in}\geq 0$ a.e., 
there exists a solution $f\in L^\infty((0,\infty); X)$, $\phi\in C^1((0,\infty))\cap W_\mathrm{loc}^{2,\infty}([0,\infty))$ of ~\eqref{FP}-\eqref{nordstrom} such that $f(0,p)=f_\mathrm{in}(p)$ and $(\phi(0),\dot{\phi}(0))=(\phi_\mathrm{in},\psi_\mathrm{in})$. Moreover $f\geq 0$ a.e., the total mass is conserved, i.e.,   $$\|f(t)\|_{L^1(\R^3)}=\|f_\mathrm{in}\|_{L^1(\R^3)},$$ and there exist constants $\alpha,\beta,\varepsilon,C>0$ such that
\begin{equation}\label{pointestfield}
-C-\alpha\, t\leq\phi(t)\leq C-\beta\,t,\quad |\dot{\phi}(t)|<C,\quad -C e^{-\alpha\, t}\leq \ddot\phi(t)\leq 0,
\end{equation}
\begin{equation}\label{fnovanish}
\mu(\{p:|f(t,p)|>\varepsilon)\})>C,
\end{equation}
where $\mu$ denotes the Lebesgue measure.
Finally, if the initial datum $f_\mathrm{in}$ satisfies 
\begin{equation}
\label{secondmombounded}
\int_{\R^3}\left[ (1 + |p|^2)^\delta |\nabla_p f_\mathrm{in}|^2 + \left( 1 + |p|^2 \right)^{\delta+1} |\nabla_p^2 f_\mathrm{in}|^2 \right] \, dp<\infty,
\end{equation}
for some $\delta>1/2$, then the estimate 
\begin{equation}
\label{momsbounded}
\int_{\mathbb{R}^3} ( 1 + |p|^2)^{\delta} \vert \nabla_p f \vert^2\,dp+(1+t)^{-1} \int_{\R^3} (1 + |p|^2)^{\delta+1} \vert \nabla^2_p f \vert^2 \, dp < C,
\end{equation} 
holds for all $t>0$, and the solution is unique.
\end{theorem}
Notice that the estimate~\eqref{fnovanish} shows that $f$ is not vanishing, not even asymptotically so. We remark that the crucial ingredient to prove~\eqref{fnovanish} is the uniform estimate $\int_{\R^3} |p|f\,dp\leq C$, which in turn is a consequence of the field decay, see~\eqref{estimatesf2}. 
As the proof of Theorem~\ref{global} is rather long, we split it into several subsections.

\subsection{The Nordstr\"om equation}
Assume first that $0\leq f\in C((0,\infty);L^1(\R^3))$ is given and consider the Cauchy problem for the Nordstr\"om field equation
\begin{subequations}\label{ns}
\begin{align} 
\ddot{\phi}(t) &=-H_f(t,\phi),\quad t>0,\label{nso}\\
\phi(0) & =\phi_\mathrm{in},\ \dot{\phi}(0)=\psi_{\mathrm{in}},\label{nosin0}
\end{align}
with 
\begin{equation}
H_f(t,\phi)=e^{2\phi}\int_{\mathbb{R}^3}\!\frac{f(t,p)}{\sqrt{e^{2\phi}+|p|^2}}\, dp . \label{ache}
\end{equation}
\end{subequations}
Since the function $x\to e^{2x}/\sqrt{e^{2x}+|p|^2}$ is convex and monotonically increasing, we have
\begin{align}
|H_f(t,\phi_1)-H_f(t,\phi_2)&|\leq\partial_\phi H_f(t,\phi_*)|\phi_2-\phi_1|= e^{2\phi_*}|\phi_2-\phi_1|\int_{\R^3}f(t,p)\frac{e^{2\phi_*}+2|p|^2}{(e^{2\phi_*}+|p|^2)^{3/2}}\,dp\nonumber\\
&\leq 2\|f(t)\|_{L^1(\R^3)}e^{\phi_*}|\phi_2-\phi_1|,\label{lip}
\end{align}
where $\phi_*=\max\{\phi_1,\phi_2\}$.
Letting $\psi=\dot{\phi}, y=(\phi,\psi)$ and $F(t,y)=(y_2,-H_f(t,y_1))$, equation~\eqref{nso} becomes $\dot{y}=F(t,y)$. From the regularity of $f$ and the estimate~\eqref{lip}, the function $F$ is continuous in $t>0$ and locally Lipschitz in $y$, uniformly in the time variable. It follows by Picard's theorem that the Cauchy problem~\eqref{ns} has a unique local classical solution. Moreover by straightforward estimates we obtain
\begin{subequations}\label{estfield}
\begin{align}
-\mathcal{K}_f(t)e^{\phi(t)} \leq -H_f(t,\phi)=\ddot{\phi}(t) &\leq 0, \\
\psi_\mathrm{in}-\mathcal{K}_f(t)\int_0^t\!e^{\phi(s)}\,ds\leq \dot{\phi}(t) &\leq \psi_\mathrm{in},\\
\psi_\mathrm{in} t+\phi_\mathrm{in}-\mathcal{K}_f(t)\int_0^t\int_0^s\!e^{\phi(\tau)}\,d\tau\,ds \leq \phi(t) & \leq \psi_\mathrm{in} t+\phi_\mathrm{in},
\end{align}
\end{subequations}
where
\begin{equation}\label{defK}
\mathcal{K}_f(t)=\sup_{s\in (0,t)}\|f(s)\|_{L^1(\R^3)}.
\end{equation}
These estimates imply that $\phi\in W^{2,\infty}((0,T))$ and
\begin{equation}\label{estphi}
\|\phi\|_{W^{2,\infty}((0,T))}\leq C_T\mathcal{K}_f(T),
\end{equation}
for all $T>0$. Hence we have proved 
\begin{proposition}\label{exiunin} The Cauchy problem~\eqref{ns} has a unique global solution $\phi\in C^2((0,\infty))$. Moreover the solution satisfies the estimates~\eqref{estfield}-\eqref{estphi}, for all $t\in [0,T]$ and $T>0$.
\end{proposition}

\subsection{The linear Fokker-Planck equation}
In this section we assume that $\phi\in C^2((0,\infty))\cap W^{1,\infty}_\mathrm{loc}([0,\infty))$ is given and consider the Cauchy problem for the linear Fokker-Planck equation
\begin{subequations}\label{cauchyFP}
\begin{align}
&\partial_t f=e^{2\phi}\partial_{p^i}(D^{ij}[\phi]\partial_{p^j}f),\quad t>0,\ p\in\R^3,\label{FPlin}\\
& f(0,p)=f_\mathrm{in}(p),\quad p\in\R^3,
\end{align}
\end{subequations}
where we recall that the diffusion matrix $D^{ij}[\phi]$ is defined by~\eqref{defD}. The purpose of this subsection is to prove the following result.
\begin{proposition}\label{existenceFP}
Given $0\leq f_\mathrm{in}\in C^2_c(\R^3)$, there exists a unique, positive, classical solution of the Cauchy problem~\eqref{cauchyFP}. Moreover $f$ satisfies
\begin{equation}\label{estimatesf}
\|f(t)\|_{L^1(\R^3)}=\|f_\mathrm{in}\|_{L^1(\R^3)},\qquad \|f(t)\|_{L^q(\R^3)}\leq \|f_\mathrm{in}\|_{L^q(\R^3)},\ \text{for all $q>1$.}
\end{equation}
Finally, for all $\gamma\geq 0$ there exists a constant $C>0$, which depends only on $\gamma$, such that
\begin{equation}\label{estimatesf2}
\int_{\R^3}(e^{2\phi}+|p|^2)^\gamma\,(f+|\nabla_pf|^2)\,dp\leq C\exp \left(C \int_0^t\mathcal{Q}_\phi(s)\,ds\right)\int_{\R^3}(e^{2\phi}+|p|^2)^\gamma\,(f_\mathrm{in}+|\nabla_pf_\mathrm{in}|^2)\,dp
\end{equation}
\begin{align}
\int_{\R^3}(e^{2\phi}+|p|^2)^\gamma\,|\nabla^2_pf|^2\,dp&\leq C\left (\int_{\R^3}(e^{2\phi}+|p|^2)^\gamma\,|\nabla^2_pf_\mathrm{in}|^2+\int_0^t \int_{\R^3}(e^{2\phi}+|p|^2)^\gamma\,|\nabla_pf|^2\,dp\ ds \right)\nonumber\\
&\qquad\qquad\times\exp \left(C \int_0^t\mathcal{Q}_\phi(s)\,ds\right)\label{seconderest}
\end{align}
where, denoting $(z)_+=\max(0,z)$,
\begin{equation}\label{defQ}
\mathcal{Q}_\phi(t)=e^{\phi(t)}+(\dot{\phi})_+(t).
\end{equation}
\end{proposition}  
\begin{proof}
For the proof of existence we employ methods from the theory of stochastic differential equations and diffusion processes described in~\cite{arnold}. To adhere with the formulation used in~\cite{arnold}, we define the functions
\[
\bar{f}(t,p)=f(-t,p),\quad \bar{\phi}(t)=\phi(-t),\quad t<0,\ p\in\R^3,
\] 
in terms of which the Cauchy problem~\eqref{cauchyFP} becomes
\begin{subequations}\label{newcauchyFP}
\begin{align}
&\partial_t\bar{f}+\mathcal{D}\bar{f}=0,\quad t<0,\ p\in\R^3,\label{newFP}\\
&\bar{f}(0,p)=f_\mathrm{in}(p)\quad p\in\R^3,
\end{align}
\end{subequations}
where $\mathcal{D}$ is the differential operator
\[
\mathcal{D}=d^i\partial_{p^i}+\frac{1}{2}b^{ij}\partial_{p^i}\partial_{p^j},
\]
and
\begin{equation}\label{ddef}
d^i(t,p)=e^{2\bar{\phi}}\partial_{p^j}(D^{ij}[\bar{\phi}])=\frac{3e^{2\bar{\phi}}p^i}{\sqrt{e^{2\bar{\phi}}+|p|^2}},\qquad b^{ij}(t,p)=2e^{2\bar{\phi}}D^{ij}[\bar{\phi}].
\end{equation}
Let $G$ denote the square root of $b$, i.e., the unique positive definite matrix such that $b=G\cdot G^T$. It can be verified that
\begin{equation}\label{Gdef}
G^{ij}=\frac{\sqrt{2}e^{\bar{\phi}}}{(e^{2\bar{\phi}} +|p|^2)^{1/4}}\left(e^{\bar{\phi}}\delta^{ij}+\frac{p^i p^j}{e^{\bar{\phi}}+\sqrt{e^{2\bar{\phi}} +|p|^2}}\right).
\end{equation}
Now let $T>0$ be fixed. Note that $|d(t,p)|+|G(t,p)|\leq C_T (1+|p|)$, for $t\in [-T,0]$. Additionally, we show in the appendix that the first and second derivatives of $d$ and $G$ with respect to $p$ are bounded. These estimates are exactly those required to apply~\cite[Th.~9.4.4]{arnold}. Hence for any $t\in [-T,0]$, the system of stochastic differential equations
\begin{equation}\label{sde}
dP=b(s,P)\,ds+G(s,P)\cdot dW,
\end{equation}
with $dW$ denoting the standard Wiener process, admits a unique solution $P(s; x,t)$, $t\leq s\leq 0$ satisfying $P(t;x,t)=p$. Moreover, the Feynman-Kac formula
\[
\bar{f}(t,p)=\mathbb{E}[f_\mathrm{in}(P(0,p;t))]
\] 
is a classical positive solution of the Cauchy problem~\eqref{newcauchyFP} in the interval $[-T,0]$. We recall that in the theory of stochastic differential equations, equation~\eqref{newFP} is the  backward Kolmogorov equation associated to the system~\eqref{sde}. In conclusion, transforming back to the original variables $(f,\phi)$, we have found a classical solution $f(t,p)$ of~\eqref{cauchyFP} defined for all $t>0$. 
Next, we show that classical solutions satisfy the estimates~\eqref{estimatesf}. Let $\xi\in C_c^\infty([0,\infty))$ be a non-increasing function such that
\begin{align*}
\xi(r) = \begin{cases}
1 & \text{if } 0 \leq r \leq 1,\\
0 & \text{if } r \geq 2,
\end{cases}
\end{align*}
and define $\xi_n(p) = \xi \left(\frac{|p|}{n}\right)$, for $p\in \R^3$ and $n\in\mathbb{N}$, $n \geq 1$. Then $\xi_n \in C^\infty_c(\R^3)$ is a cut-off function satisfying $0 \leq \xi_n \leq 1$, $\xi_n(p) = 1$ if $|p| \leq n$, and $\xi_n(p) = 0$ if $|p| \geq 2n$. We clearly have $|\nabla _p \xi_n| \leq C/n$ and $|\Delta _p \xi_n| \leq C/n^2$. By direct computation we obtain
\begin{align}
\frac{d}{dt}\int_{\R^3}\xi_nf^q\,dp=&-q(q-1)e^{2\phi}\int_{\R^3}\xi_n f^{q-2}D^{ij}[\phi]\partial_{p^i}f\partial_{p^j}f\,dp\nonumber\\
&+e^{2\phi}\int_{\R^3}f^q\left[(\partial_{p^j}D^{ij}[\phi])\partial_{p^i}\xi_n+D^{ij}[\phi]\partial_{p^i}\partial_{p^j}\xi_n\right]dp,\label{identity}
\end{align}
for all $q\geq 1$. By the positivity of $D$, the first term in the right side of~\eqref{identity} is non-positive. From the properties of the cutoff function, the term in square brackets in the last integral satisfies
\[
[\dots]\leq \frac{C_T}{n},\quad \text{for all $t\in [0,T]$ and all $T>0$.}
\]
Hence using Gr\"onwall's inequality,  the identity~\eqref{identity} gives
\[
\|f(t)\|_{L^q(\R^3)}\leq C_T.
\]
Substituting again in~\eqref{identity} we get the inequalities
\[
\|f_\mathrm{in}\|_{L^1(\R^3)}-\frac{C_T}{n}\leq \|f(t)\|_{L^1(\R^3)}\leq \|f_\mathrm{in}\|_{L^1(\R^3)}+\frac{C_T}{n},\quad \|f(t)\|_{L^q(\R^3)}\leq \|f_\mathrm{in}\|_{L^q(\R^2)}+\frac{C_T}{n}.
\]
Taking the limit as $n\to\infty$ proves~\eqref{estimatesf}. The uniqueness statement of the proposition follows by applying the estimate on $\|f(t)\|_{L^2(\R^3)}$ to the difference of two solutions. As to the estimates~\eqref{estimatesf2}-\eqref{seconderest}, we present for brevity only a formal proof; the computations can be made rigorous by introducing the cut-off function $\xi$ as above. We then compute 
\begin{align*}
\frac{d}{dt}\int_{\R^3}(e^{2\phi}+|p| ^2)^\gamma\,f\,dp&=2\gamma e^{2\phi}\dot\phi\int_{\R^3}(e^{2\phi}+|p|^2)^{\gamma-1}\,f\,dp+e^{2\phi}\int_{\R^3}(e^{2\phi}+|p|^2)^\gamma\,\partial_{p^i}(D^{ij}[\phi]\partial_{p^j}f)\,dp\\
&\leq C(\dot{\phi})_+\int_{\R^3}(e^{2\phi}+|p|^2)^\gamma\,f\,dp+e^{2\phi}\int_{\R^3}\partial_{p^j}\big\{D^{ij}[\phi]\partial_{p^i}[(e^{2\phi}+|p|^2)^\gamma]\big\}\,f\,dp,
\end{align*}
where $(\cdot)_+$ denotes the positive part. Bounding the bracketed portion $\{...\}$ of the second term, we find 
\[
\nabla_p\cdot\{\dots\}
=4\gamma(\gamma-1/2)(e^{2\phi}+|p|^2)^{\gamma-3/2}|p|^2+6\gamma(e^{2\phi}+|p|^2)^{\gamma-1/2}\leq C e^{-\phi}(e^{2\phi}+|p|^2)^{\gamma},
\]
and thus obtain
\[
\frac{d}{dt}\int_{\R^3}(e^{2\phi}+|p| ^2)^\gamma\,f\,dp\leq C (e^{\phi}+(\dot{\phi})_+)\int_{\R^3}(e^{2\phi}+|p| ^2)^\gamma\,f\,dp,
\]
which by Gr\"onwall's inequality implies
\[
\int_{\R^3}(e^{2\phi}+|p| ^2)^\gamma\,f\,dp\leq \exp \left(C \int_0^{t}\mathcal{Q}_\phi(s)\,ds\right) \int_{\R^3}(e^{2\phi}+|p| ^2)^\gamma\,f_\mathrm{in}\,dp.
\]
As to the estimate on $\nabla_p f$, we compute
\begin{align}
\frac{d}{dt}\int_{\R^3} (e^{2\phi}+|p|^2)^\gamma\,|\nabla_pf|^2\,dp=&\ 2\gamma e^{2\phi}\dot{\phi}\int_{\R^3}(e^{2\phi}+|p|^2)^{\gamma-1}|\nabla_pf|^2\,dp\nonumber\\
&+\underbrace{2e^{2\phi}\int_{\R^3}(e^{2\phi}+|p|^2)^\gamma\nabla_pf\cdot\nabla_p(\partial_{p^i}(D^{ij}[\phi]\partial_{p^j}f))\,dp}_{(*)}.\label{temporale}
\end{align}
In the integral in $(*)$ we first integrate by parts in the variable $p^i$ and then, after straightforward calculations, we obtain 
\[
(*)=I_1+I_2+I_3+I_4,
\]
where
\begin{align*}
&I_1=-2e^{2\phi}\int_{\R^3}(e^{2\phi}+|p|^2)^\gamma D^{ij}[\phi]\partial_{p^i}\nabla_pf\cdot\partial_{p^j}\nabla_pf\,dp\\
&I_2=-4\gamma e^{2\phi}\int_{\R^3}(e^{2\phi}+|p|^2)^{\gamma} A^{jk}[\phi]\partial_{p^j}f\,\partial_{p^k}f\,dp\\
&I_3=2\gamma e^{2\phi}\int_{\R^3}\nabla_p\cdot(p\,(e^{2\phi}+|p|^2)^{\gamma-1/2})|\nabla_p f|^2
\,dp\\
&I_4=-2e^{2\phi}\int_{\R^3}(e^{2\phi}+|p|^2)^\gamma\partial_{p^i}\partial_{p^k}f(\partial_{p_k}D^{ij}[\phi])\partial_{p^j}f\,dp
\end{align*}
and
\begin{equation}\label{matrixA}
A^{j}_{\ k}[\phi]=\frac{p_i\partial_{p^k}D^{ij}[\phi]}{e^{2\phi}+|p|^2}=\frac{\delta^{j}_{\ k}|p|^2}{(e^{2\phi}+|p|^2)^{3/2}}.
\end{equation}
By the positivity of $D$ and $A$ we have $I_1+I_2\leq 0$. 
In the integral $I_3$ we compute
\[
\nabla_p\cdot(p\,(e^{2\phi}+|p|^2)^{\gamma-1/2})=3(e^{2\phi}+|p|^2)^{\gamma-1/2}+(2\gamma-1)(e^{2\phi}+|p|^2)^{\gamma-3/2}|p|^2\leq C  e^{-\phi}(e^{2\phi}+|p|^2)^\gamma
\]
and thus the integral $I_3$ is bounded by
\[
I_3\leq C e^{\phi}\int_{\R^3} (e^{2\phi}+|p|^2)^\gamma\,|\nabla_pf|^2\,dp.
\]
The integral $I_4$ requires some further work. Integrating by parts the $p^k$ derivative in $\partial_{p^k} f$, we obtain
\begin{align*}
I_4=\ &2e^{2\phi}\int_{\R^3}(e^{2\phi}+|p|^2)^\gamma (\Delta_p D^{ij}[\phi])\partial_{p^i}f\partial_{p^j}f\,dp\\
&+4\gamma e^{2\phi}\int_{\R^3}(e^{2\phi}+|p|^2)^\gamma B^{ij}[\phi]\partial_{p^i}f\,\partial_{p^j}f\,dp\\
&+2e^{2\phi}\int_{\R^3}(e^{2\phi}+|p|^2)^\gamma\partial_{p^i}f(\partial_{p_k}D^{ij}[\phi])\partial_{p^j}\partial_{p^k}f\,dp,
\end{align*}
where 
\begin{equation}\label{matrixB}
B^{ij}[\phi]=\frac{p\cdot\nabla_pD^{ij}[\phi]}{(e^{2\phi}+|p|^2)}=\frac{2p^ip^j}{(e^{2\phi}+|p|^2)^{3/2}}-\frac{|p|^2}{(e^{2\phi}+|p|^2)^{2}}D^{ij}[\phi],
\end{equation}
\[
\Delta_pD^{ij}[\phi]=\frac{1}{\sqrt{e^{2\phi}+|p|^2}}\left(2\delta^{ij}-\frac{4p^ip^j}{e^{2\phi}+|p|^2}\right)-\frac{3e^{2\phi}}{(e^{2\phi}+|p|^2)^2}D^{ij}[\phi].
\]
By the symmetry of $D$, the last integral equals $-I_4$ and thus we have obtained
\[
I_4=I_{4A}+I_{4B},
\]
where 
\[
\begin{gathered}
I_{4A}=e^{2\phi}\int_{\R^3}(e^{2\phi}+|p|^2)^\gamma (\Delta_p D^{ij}[\phi])\partial_{p^i}f\partial_{p^j}f\,dp\\
 I_{4B}=2\gamma e^{2\phi}\int_{\R^3}(e^{2\phi}+|p|^2)^\gamma B^{ij}[\phi]\partial_{p^i}f\,\partial_{p^j}f\,dp.
\end{gathered}
\]
Using the straightforward bounds 
\begin{equation}\label{estmatrix}
B^{ij}[\phi]x_ix_j\leq Ce^{-\phi}|x|^2,\qquad \Delta_pD^{ij}[\phi]x_ix_j\leq Ce^{-\phi}|x|^2,\ \text{for all $x\in\R^3$},
\end{equation}
we have
\[ 
I_4\leq Ce^{\phi}\int_{\R^3} (e^{2\phi}+|p|^2)^\gamma\,|\nabla_pf|^2\,dp.
\]
Collecting the estimates we find that the term $(*)$ in~\eqref{temporale} satisfies
\begin{equation}\label{est*}
(*)\leq  Ce^{\phi}\int_{\R^3} (e^{2\phi}+|p|^2)^\gamma\,|\nabla_pf|^2\,dp,
\end{equation}
which yields
\[
\frac{d}{dt}\int_{\R^3} (e^{2\phi}+|p|^2)^\gamma\,|\nabla_pf|^2\,dp\leq C(e^{\phi}+(\dot\phi)_+)\int_{\R^3} (e^{2\phi}+|p|^2)^\gamma\,|\nabla_pf|^2\,dp.
\]
An application of Gr\"onwall's inequality completes the proof of~\eqref{estimatesf2}. 

To prove~\eqref{seconderest} we use that $g_k=\partial_{p^k}f$ satisfies, for all $k=1,2,3$, 
\[
\partial_t g_k=e^{2\phi} \partial_{p^i}(D^{ij}\partial_{p^j}g_k)+e^{2\phi}\partial_{p^i}[(\partial_{p^k}D^{ij})g_j]
\]
and thus
\begin{align}
\frac{d}{dt}\int_{\R^3}(e^{2\phi}+|p|^2)^\gamma\nabla_p g^k\cdot\nabla_p g_k&=2\gamma\int_{\R^3}(e^{2\phi}+|p|^2)^{\gamma-1}e^{2\phi}\dot\phi\,\nabla_pg^k\cdot\nabla_pg_k\,dp\nonumber\\
&\quad +2e^{2\phi}\int_{\R^3}(e^{2\phi}+|p|^2)^\gamma\,\nabla_p g^k\cdot \nabla_p[\partial_{p^i}(D^{ij}\partial_{p^j}g_k)]\,dp\nonumber\\
&\quad +2 e^{2\phi}\int_{\R^3}(e^{2\phi}+|p|^2)^\gamma\,\nabla_p g^k\cdot\nabla_p\{\partial_{p^i}[(\partial_{p^k}D^{ij})g_j]\}\,dp\nonumber\\
& =II+III+IV.\label{seconderL2}
\end{align}
(In the rest of the proof we denote $D^{ij}[\phi]=D^{ij}$ for notational simplicity.)
The term $III$ is the same as $(*)$ in~\eqref{temporale} with $f$ replaced by $g_k$ and thus by~\eqref{est*} it satisfies the bound
\begin{equation}\label{estII}
III\leq C e^\phi\int_{\R^3}(e^{2\phi}+|p|^2)^\gamma\,\nabla_p g^k\cdot \nabla_p g_k\,dp\leq Ce^\phi\int_{\R^3}(e^{2\phi}+|p|^2)^\gamma|\nabla_p^2f|^2\,dp.
\end{equation} 
The term $II$ is bounded as
\begin{equation}\label{estsimple}
II\leq C(\dot\phi)_+\int_{\R^3}(e^{2\phi}+|p|^2)^\gamma |\nabla_p^2f|^2\,dp.
\end{equation}
Expanding the term $IV$ in~\eqref{seconderL2} we obtain
\begin{align}
IV&=2e^{2\phi}\int_{\R^3}(e^{2\phi}+|p|^2)^\gamma\,\nabla_pg^k\cdot\nabla_p(\partial_{p^i}\partial_{p^k}D^{ij})g_j\,dp\nonumber\\
&\quad +2e^{2\phi}\int_{\R^3}(e^{2\phi}+|p|^2)^\gamma(\partial_{p^i}\partial_{p^k}D^{ij})\nabla_p g^k\cdot\nabla_p g_j\,dp\nonumber\\
&\quad+2e^{2\phi}\int_{\R^3}(e^{2\phi}+|p|^2)^\gamma\,\nabla_pg^k\cdot\nabla_p(\partial_{p^k}D^{ij})\partial_{p^i}g_j\,dp\nonumber\\
&\quad+2e^{2\phi}\int_{\R^3}(e^{2\phi}+|p|^2)^\gamma(\partial_{p^k}D^{ij})\nabla_p g^k\cdot\nabla_p\partial_{p^i}g_j\,dp\nonumber\\
&=IV_1+IV_2+IV_3+IV_4.\label{estIII}
\end{align}
In $IV_4$ we integrate by parts in the $p^i$ derivative acting on $g_j$ and obtain
\begin{align}
IV_4&=-4\gamma e^{2\phi}\int_{\R^3}(e^{2\phi}+|p|^2)^{\gamma-1}(p_i\partial_{p^k}D^{ij})\nabla_p g^k\cdot\nabla_pg_j\,dp\nonumber\\
&\quad-2e^{2\phi}\int_{\R^3}(e^{2\phi}+|p|^2)^\gamma(\partial_{p^i}\partial_{p^k}D^{ij})\nabla_p g^k\cdot\nabla_pg_j\,dp\\
&\quad-2 e^{2\phi}\int_{\R^3}(e^{2\phi}+|p|^2)^\gamma(\partial_{p^k}D^{ij})\nabla_p\partial_{p^i}g^k\cdot\nabla_p g_j\,dp\nonumber\\
&=IV_{4A}+IV_{4B}+IV_{4C}.\label{III4}
\end{align}
Note that $IV_2+IV_{4B}=0$.  In $IV_{4C}$ we integrate by parts in the $p^k$  derivative within $g_k=\partial_{p^k}f$, so that
\begin{align*}
IV_{4C}&=4\gamma e^{2\phi}\int_{\R^3}(e^{2\phi}+|p|^2)^{\gamma-1} (p^k\partial_{p^k}D^{ij})\nabla_pg_i\cdot\nabla_pg_j\,dp\\
&\quad +2e^{2\phi}\int_{\R^3}(e^{2\phi}+|p|^2)^\gamma(\Delta_pD^{ij})\nabla_p g_i\cdot\nabla_p g_j
\,dp\\
&\quad+2e^{2\phi}\int_{\R^3}(e^{2\phi}+|p|^2)^\gamma(\partial_{p^k}D^{ij})\nabla_p g_i\cdot\nabla_p\partial_{p^j}g^k\,dp.
\end{align*}
By the symmetry of $D$, the third term in the right hand side of the latter equation equals $-IV_{4C}$, hence
\begin{align}
IV_{4C}&=e^{2\phi}\int_{\R^3}(e^{2\phi}+|p|^2)^\gamma(\Delta_pD^{ij})\nabla_p g_i\cdot\nabla_p g_j\nonumber\\
&\quad+2\gamma e^{2\phi}\int_{\R^3}(e^{2\phi}+|p|^2)^{\gamma-1}(p^k\partial_{p^k}D^{ij})\nabla_pg_i\cdot\nabla_pg_j\,dp.\label{III4B}
\end{align} 
Substituting~\eqref{III4B} into~\eqref{III4} and then returning to~\eqref{estIII} we obtain
\begin{align}
IV&=2e^{2\phi}\int_{\R^3}(e^{2\phi}+|p|^2)^\gamma\nabla_p(\partial_{p^i}\partial_{p^k}D^{ij})\cdot (\nabla_pg^k)g_j\,dp\nonumber\\
&\quad+2e^{2\phi}\int_{\R^3}(e^{2\phi}+|p|^2)^\gamma\nabla_p(\partial_{p^k}D^{ij})\cdot\nabla_pg^k\partial_{p^i}g_j\,dp\nonumber\\
&\quad+e^{2\phi}\int_{\R^3}(e^{2\phi}+|p|^2)^\gamma(\Delta_pD^{ij})\nabla_pg_i\cdot\nabla_pg_j\,dp\nonumber\\
&\quad -4\gamma e^{2\phi}\int_{\R^3}(e^{2\phi}+|p|^2)^{\gamma}A^j_{\ k}\nabla_pg^k\cdot\nabla_pg_j\,dp\nonumber\\
&\quad +2\gamma e^{2\phi}\int_{\R^3}(e^{2\phi}+|p|^2)^{\gamma}B^{ij}\nabla_pg_i\cdot\nabla_pg_j\,dp,\label{estIIInew}
\end{align}
where $A^j_{\ k}$ and $B^{ij}$ are given by~\eqref{matrixA},~\eqref{matrixB}. Recall that $A$ is positive definite and that the estimates~\eqref{estmatrix} hold. Furthermore, we show in the appendix that
\begin{equation}\label{boundsderD}
|\nabla_p\partial_{p^k}D^{ij}|\leq C e^{-\phi},\quad |\nabla_p(\partial_{p^i}\partial_{p^k}D^{ij})|\leq Ce^{-2\phi}.
\end{equation}
With these bounds, equation~\eqref{estIIInew} entails
\begin{align}\label{estIIIfin}
IV&\leq C\int_{\R^3}(e^{2\phi}+|p|^2)^{\gamma}|\nabla_pg_k|\,|g_j|\,dp+C e^{\phi}\int_{\R^3}(e^{2\phi}+|p|^2)^{\gamma}(|\nabla_p g_k|+|\nabla_p g_i|)|\nabla_p g_j|\,dp\nonumber\\
&\leq C\left(\int_{\R^3}(e^{2\phi}+|p|^2)^\gamma|\nabla_p^2f|^2\,dp\right)^{1/2}\left(\int_{\R^3}(e^{2\phi}+|p|^2)^\gamma|\nabla_pf|^2\,dp\right)^{1/2}\nonumber\\
&\quad+e^\phi\int_{\R^3}(e^{2\phi}+|p|^2)^\gamma|\nabla_p^2f|^2\,dp,
\end{align}
where we used H\"older's inequality in the last step. 
Substituting the bounds~\eqref{estII},~\eqref{estsimple} and~\eqref{estIIIfin} into~\eqref{seconderL2} we obtain
\begin{align*}
\frac{d}{dt}\int_{\R^3}(e^{2\phi}+|p|^2)^\gamma|\nabla_p^2f|^2\,dp&\leq C(e^\phi+(\dot\phi)_+)\int_{\R^3}(e^{2\phi}+|p|^2)^\gamma|\nabla_p^2f|^2\,dp\\
&\quad+ C\left(\int_{\R^3}(e^{2\phi}+|p|^2)^\gamma|\nabla_p^2f|^2\,dp\right)^{1/2}\left(\int_{\R^3}(e^{2\phi}+|p|^2)^\gamma|\nabla_pf|^2\,dp\right)^{1/2}
\end{align*}
and therefore 
\begin{align*}
\frac{d}{dt}\left(\int_{\R^3}(e^{2\phi}+|p|^2)^\gamma|\nabla_p^2f|^2\,dp\right)^{1/2}&\leq C(e^\phi+(\dot\phi)_+)\left(\int_{\R^3}(e^{2\phi}+|p|^2)^\gamma|\nabla_p^2f|^2\,dp\right)^{1/2}\\
&\quad+C\left(\int_{\R^3}(e^{2\phi}+|p|^2)^\gamma|\nabla_pf|^2\,dp\right)^{1/2},
\end{align*}
which by Gr\"onwall's inequality gives~\eqref{seconderest}. 
\end{proof}

\subsection{Existence}\label{exuni}
By a simple density argument we can assume that $f_\mathrm{in}\in C^2_c(\R^3)$. We fix $T>0$ and consider the sequence $(f_n,\phi_n)$ defined iteratively as follows. For $n=0$ we set $(f_0,\phi_0)=(f_\mathrm{in},\phi_\mathrm{in})$. Assuming that the pair $(f_n,\phi_n)$ is given, we define $(f_{n+1},\phi_{n+1})$ as the unique solution of the equations
\[
\partial_tf_{n+1}=e^{2\phi_n}\partial_{p^i}(D^{ij}[\phi_n]\partial_{p^j}f_{n+1}),\qquad \ddot\phi_{n+1}=-e^{2\phi_{n+1}}\int_{\R^3}\frac{f_{n+1}}{\sqrt{e^{2\phi_{n+1}}+|p|^2}}\,dp,
\]  
with initial data $f_{n+1}(0,p)=f_\mathrm{in}(p)$, $(\phi_{n+1}(0),\dot{\phi}_{n+1}(0))=(\phi_\mathrm{in},\psi_\mathrm{in})$. It follows by induction and Propositions~\ref{exiunin} and \ref{existenceFP} that the sequence $(f_n,\phi_n)$ consists of smooth functions. Moreover, by~\eqref{estimatesf},
\[
\|f_n(t)\|_{L^1(\R^3)}=\|f_\mathrm{in}\|_{L^1(\R^3)},\quad \|f_n(t)\|_{L^2(\R^3)}\leq\|f_\mathrm{in}\|_{L^2(\R^3)}
\]
and the function $\mathcal{K}_{f_n}(t)$ given by~\eqref{defK} is equibounded along the sequence $f_n$. Thus, by~\eqref{estphi}, 
\[
\|\phi_n\|_{W^{2,\infty}((0,T))}\leq C_T.
\]
We infer that  the function $\mathcal{Q}_{\phi_n}(t)$ given by~\eqref{defQ} is equibounded along the sequence $\phi_n$. Hence, by~\eqref{estimatesf2},
\[
\|\nabla_pf_n(t)\|_{L^2(\R^3)}+\int_{\R^3}|p|f_n\,dp\leq C_T,\quad \text{for all $t\in[0,T]$}.
\]
It follows that there exist 
\[
f\in L^\infty((0,T); H^1(\R^3)),\qquad \phi\in W^{2,\infty}((0,T))
\]
and a subsequence, still denoted $(f_n,\phi_n)$, such that
\[
f_n\rightharpoonup f,\ \text{in } L^2((0,T)\times\R^3),\qquad \phi_n\overset{*}{\rightharpoonup}\phi,\ \text{in } W^{2,\infty}(0,T),\quad \text{as $n\to\infty$.} 
\]
By a standard diagonal sequence argument, we can choose $(f_n,\phi_n)$ to be independent of $T>0$. Moreover
\[
f_n(t,\cdot)\rightharpoonup f(t,\cdot)\ \text{in } H^1(\R^3) \quad \text{for all $t\in [0,T]$}.
\]
By compactness, we may extract a subsequence such that $f_n(t,\cdot)$ converges strongly in $L^2(\R^3)$ and $(\phi_n,\dot\phi_n)$ converges uniformly on $[0,T]$ (which implies in particular that $\phi\in C^1$). It is clear that this convergence is strong enough to pass to the limit in the equations and conclude that $(f,\phi)$ is a solution of the spatially homogeneous VNFP system~\eqref{FPhom}-\eqref{nordstromhom}. Moreover $f\in L^{\infty}((0,\infty);L^1(\R^3)\cap L^2(\R^3))$ and it is easy to show that $f_n(t,\cdot)\rightharpoonup f(t,\cdot)$ in $L^1(\R^3)$ (up to subsequences) so that, in particular, the mass of $f$ is preserved. In fact, the sequence $f_n$ is bounded in $L^2(\R^3)$ and it is tight, because  $|p|\,f_n$ is bounded in $L^1(\R^3)$. Hence weak convergence in $L^1(\R^3)$ of $f_n(t,\cdot)$ follows by the Dunford-Pettis theorem.

\subsection{Uniform estimates and asymptotic behavior of the field}
Next we show that $|p|f\in L^\infty((0,\infty); L^1(\R^3))$ and $\nabla_pf\in L^\infty((0,\infty); L^2(\R^2))$. Moreover, we establish the estimate~\eqref{pointestfield}. We also prove that~\eqref{momsbounded} holds when the initial data satisfy~\eqref{secondmombounded}. To this purpose we first notice that since $\dot{\phi}$ is decreasing, the limit 
\[
\dot\phi_\infty=\lim_{t \to \infty}\dot\phi(t)
\]
exists. 
\begin{lemma}\label{limitphidot}
$\dot{\phi}_\infty<0.$
\end{lemma}
\begin{proof}
Let 
\begin{equation}\label{esttemp}
 M=\|f(t)\|_{L^1(\R^3)}=\|f_\mathrm{in}\|_{L^1(\R^3)}, \quad\mathcal{E}(t)=\int_{\R^3} f\sqrt{e^{2\phi}+|p|^2}\,dp+\frac{1}{2}\dot{\phi}^2
\end{equation}
be the mass and the energy of the solution constructed in Section~\ref{exuni}.
By H\"older's inequality,
\begin{equation}
\label{MEineq}
M^2\leq\left(\int_{\R^3}\frac{f}{\sqrt{e^{2\phi}+|p|^2}}\,dp\right)\left(\int_{\R^3} f\sqrt{e^{2\phi}+|p|^2}\,dp\right)\leq \left(\int_{\R^3}\frac{f}{\sqrt{e^{2\phi}+|p|^2}}\,dp\right)\mathcal{E}(t).
\end{equation}
Now, by a direct formal computation we have
\[
\dot{\mathcal{E}}(t)=3e^{2\phi}\int_{\R^3}f\,dp,
\]
whence
\begin{equation}\label{enid}
\mathcal{E}(t)= \mathcal{E}(0)+3M\int_0^t e^{2\phi(s)}\,ds.
\end{equation}
The previous identity holds for the solution constructed in the previous section, as it follows by applying the above formal calculation to the sequence $(f_n,\phi_n)$ and then passing to the (strong) limit resulting as $n\to\infty$. 
Using~\eqref{enid} in~\eqref{MEineq}, we arrive at 
\[
\int_{\R^3} \frac{f}{\sqrt{e^{2\phi}+|p|^2}}\,dp\geq \frac{M^2}{\mathcal{E}(0)+3M\int_0^te^{2\phi(s)}ds}.
\]
Utilizing this inequality yields
\begin{align*}
\ddot{\phi}&=-e^{2\phi}\int_{\R^3}\frac{f}{\sqrt{e^{2\phi}+|p|^2}}\,dp\leq -\frac{M^2e^{2\phi}}{\mathcal{E}(0)+3M\int_0^te^{2\phi(s)}ds}\\
&=-\frac{M}{3}\frac{d}{dt}\log\left[\mathcal{E}(0)+3M\int_0^te^{2\phi(s)}ds\right].
\end{align*}
Whence
\begin{equation}\label{estphidot}
\dot{\phi}(t)\leq \dot{\phi}(0)-\frac{M}{3}\log\left[\mathcal{E}(0)+3M\int_0^te^{2\phi(s)}ds\right]+\frac{M}{3}\log{\mathcal{E}(0)}.
\end{equation}
If $\dot{\phi}_\infty\geq 0$, then $\dot{\phi}$ is positive for all $t\in [0,\infty)$. Hence the right side of~\eqref{estphidot} tends to $-\infty$ as $t\to\infty$, a contradiction. Thus $\dot{\phi}_\infty<0$ must hold. 
\end{proof}
The previous lemma easily yields the desired estimates. In fact, since $\dot\phi_\infty<0$ and $\dot\phi$ is decreasing, there exists $t_0\geq 0$ such that $\dot{\phi}(t)<\dot{\phi}(t_0)<0$, for all $t\geq t_0$. Hence
\[
\phi(t)=\phi(t_0)+\int_{t_0}^t\dot{\phi}(s)\,ds\leq \left ( \phi(t_0) + \vert \dot{\phi}(t_0) \vert t_0 \right ) - \vert \dot{\phi}(t_0) \vert t,
\] 
and therefore $\phi(t)\leq C-\beta\,t$ holds for some $\beta,C>0$. Using this within~\eqref{estfield} yields~\eqref{pointestfield}. Finally, since $\mathcal{Q}_\phi(t)=e^{\phi(t)}$, for $t \geq t_0$, we have
\[
\int_0^\infty \mathcal{Q}_\phi(t)\,dt\leq C\left(1+\int_{t_0}^\infty e^{-\beta\,t}\,dt\right)<C,
\]
and thus the estimates $|p|f\in L^\infty((0,\infty); L^1(\R^3))$, $\nabla_pf\in L^\infty((0,\infty);L^2(\R^3))$ and \eqref{momsbounded} follow from~\eqref{estimatesf2}-\eqref{seconderest}. 

\subsection{Non-vanishing property}
Since $f$ is uniformly bounded in $L^1(\R^3)\cap L^2(\R^3)$, the estimate~\eqref{fnovanish} follows if we prove that 
\begin{equation}\label{claimtoprove}
\|f(t)\|_{L^q(\R^3)}\geq C, \quad\text{for some $q\in (1,2)$}.
\end{equation} 
This is the so called $p,q,r$-Theorem (see~\cite{LL}). To establish~\eqref{claimtoprove} we first note that for all $R>0$
\begin{align*}
0<M&=\int_{\R^3}f\,dp\leq \int_{|p|\leq R}f\,dp+\frac{1}{R}\int_{|p|\geq R}|p|f\,dp\\
&\leq (4\pi)^{1-\frac{1}{q}}\|f(t)\|_{L^q(\R^3)}R^{3-\frac{3}{q}}+\frac{1}{R}\int_{\R^3}|p|f\,dp.
\end{align*}
Optimize the previous inequality, we choose
\[
R=\left[\frac{(4\pi)^{\frac{1}{q}-1}\int_{\R^3}|p|f\,dp}{(3-\frac{3}{q})\|f(t)\|_{L^q(\R^3)}}\right]^{\frac{q}{4q-3}}
\] 
and by doing so we obtain the estimate
\[
M\leq C\|f(t)\|_{L^q(\R^3)}^{\frac{q}{4q-3}}\left(\int_{\R^3}|p|f\,dp\right)^{\frac{3(q-1)}{4q-3}}.
\] 
Because $\int_{\R^3}|p|f\,dp\leq C$,~\eqref{claimtoprove} follows.

\subsection{Uniqueness}
Finally, we prove the uniqueness statement of Theorem~\ref{global}. We do so by deriving a homogenous Gr\"onwall's type inequality on the difference of two solutions with the same initial data. For brevity we limit ourself to a formal derivation assuming all the regularity of solutions necessary for the computations which follow.  However, after regularizing with a mollifying test function $\xi \in C_c^\infty((0,T)\times \R^3)$ of the form $\xi(t,p) = \theta(t) \mu(p)$ for an appropriate choice of $\theta$ and $\mu$, one may work with only the proven regularity of solutions and make the proof completely rigorous (an example of an application of this argument can be found for instance in~\cite{bdol}). Let $(f_i,\phi_i)$, $i=1,2$, be two regular solutions with the same initial data. We let $\delta > 1/2$ be given as in the statement of the theorem and compute
\begin{align*}
\frac{d}{dt}\int_{\R^3} (1+ | p|^2 )^\delta(f_1-f_2)^2\,dp&=2\int_{\R^3}(1+ | p|^2 )^\delta (f_1-f_2) \\
&\quad \times \big[e^{2\phi_1}\partial_{p^i}(D^{ij}[\phi_1]\partial_{p^j}f_1)-e^{2\phi_2}\partial_{p^i}(D^{ij}[\phi_2]\partial_{p^j}f_2)\big]\,dp\\
&=2e^{2\phi_1}\int_{\R^3}(1+ | p|^2 )^\delta(f_1-f_2)\partial_{p^i}(D^{ij}[\phi_1]\partial_{p^j}(f_1-f_2))\,dp\\
&\quad +2\int_{\R^3}(1+ | p|^2 )^\delta(f_1-f_2)\partial_{p^i}\big[(e^{2\phi_1}D^{ij}[\phi_1]-e^{2\phi_2}D^{ij}[\phi_2])\partial_{p^j}f_2\big]\,dp\\
&=I_1+I_2.
\end{align*}
Integrating by parts and using the positivity of $D$, the first integral satisfies
\begin{align*}
I_1&=-2e^{2\phi_1}\int_{\R^3}(1+ | p|^2 )^\delta D^{ij}[\phi_1]\partial_{p^i}(f_1-f_2)\partial_{p^j}(f_1-f_2)\,dp\\
&\quad -2\delta e^{2\phi_1}\int_{\R^3}(1+ | p|^2 )^{\delta-1}p_i(f_1-f_2) D^{ij}[\phi_1]\partial_{p^j}(f_1-f_2)\,dp\\
& \leq -\delta e^{2\phi_1}\int_{\R^3}(1+ | p|^2 )^{\delta-1}p_i  D^{ij}[\phi_1]\partial_{p^j}\left [(f_1-f_2)^2 \right ]\,dp.
\end{align*}
Finally, integrating by parts again and using the properties of $D^{ij}$ (see the appendix), we find
$$I_1 \leq C_T \Vert (1 + |p|^2)^\frac{\delta}{2} (f_1 - f_2)(t) \Vert_{L^2(\mathbb{R}^3)}^2.$$
In the second integral we apply the bounds
\begin{equation}
\label{Ddiff}
|e^{2\phi_1}D^{ij}[\phi_1]-e^{2\phi_2}D^{ij}[\phi_2]|\leq C_T\sqrt{1+|p|^2}\,|\phi_1-\phi_2|,
\end{equation}
and
\begin{equation}
\label{Ddiffderiv}
\left |\partial_{p^i}\left ( e^{2\phi_1}D^{ij}[\phi_1]-e^{2\phi_2}D^{ij}[\phi_2] \right )\right |\leq C_T\,|\phi_1-\phi_2|,
\end{equation}
which follow by straightforward estimates (see the appendix). With this, we find
\begin{align*}
I_2& \leq C_T | \phi_1 - \phi_2 | \int_{\R^3} (1+ | p|^2 )^\delta |f_1-f_2| \left ( |\nabla_p f_2 | + \sqrt{1+ | p|^2}  |\nabla^2_p f_2 | \right ) \,dp\\
&\leq C_T|\phi_1-\phi_2| \left (\int_{\R^3}(1+|p|^2)^\delta |f_1-f_2|^2 \, dp \right )^{1/2}\\
&\quad \times  \left [ \left (\int_{\R^3}(1+|p|^2)^\delta |\nabla_p f_2 |^2 \, dp \right )^{1/2} + \left (\int_{\R^3}(1+|p|^2)^{\delta+1}  |\nabla^2_p f_2 |^2 \, dp\right )^{1/2} \right ] \\
& \leq C_T|\phi_1-\phi_2| \cdot  \Vert (1+|p|^2)^\frac{\delta}{2} (f_1-f_2)(t) \Vert_{L^2(\mathbb{R}^3)}\\
& \leq C_T \left (|\phi_1-\phi_2|^2 +  \Vert (1+|p|^2)^\frac{\delta}{2} (f_1-f_2)(t) \Vert^2_{L^2(\mathbb{R}^3)} \right )
\end{align*}
where we used the fact that $(1+|p|^2)^{\delta/2} \nabla_p f_2\in L^{\infty}((0,T); L^2(\R^3))$ and $(1+|p|^2)^{(\delta+1)/2} \nabla^2_p f_2\in L^{\infty}((0,T); L^2(\R^3))$, as well as  H\"older's and Young's inequality. Combining $I_1$ and $I_2$, we have the estimate
$$\| (1 + | p |^2)^{\frac{\delta}{2}} (f_1-f_2)(t)\|_{L^2(\R^3)}^2\leq C_T \left ( \int_0^t \| (1 + | p |^2)^{\frac{\delta}{2}} (f_1-f_2)(s)\|_{L^2(\R^3)}^2 \, ds +  \|\phi_1-\phi_2\|^2_{L^\infty((0,t))} \right ).$$
Invoking Gr\"{o}nwall's inequality and taking the square root, we find
\begin{equation}\label{deltaf}
\| (1 + | p |^2)^{\frac{\delta}{2}} (f_1-f_2)(t)\|_{L^2(\R^3)} \leq C_T  \|\phi_1-\phi_2\|_{L^\infty((0,t))}.
\end{equation}
Next, we compute 
\begin{align*}
\phi_1-\phi_2&=-\int_0^t\int_0^s\int_{\R^3}\left(\frac{f_1e^{2\phi_1}}{\sqrt{e^{2\phi_1}+|p|^2}}-\frac{f_2e^{2\phi_2}}{\sqrt{e^{2\phi_2}+|p|^2}}\right)dp\,d\tau\,ds\\
&=\int_0^t\int_0^s\int_{\R^3}f_1\left(\frac{e^{2\phi_1}}{\sqrt{e^{2\phi_1}+|p|^2}}-\frac{e^{2\phi_2}}{\sqrt{e^{2\phi_2}+|p|^2}}\right)dp\,d\tau\,ds\\
& \quad + \  \int_0^t\int_0^s\int_{\R^3}\frac{e^{2\phi_2}}{\sqrt{e^{2\phi_2}+|p|^2}}(f_1-f_2)\,dp\,d\tau\,ds\\
&=I_3+I_4.
\end{align*}
In the first integral we simply use the Mean Value Theorem so that
\[
\left|\frac{e^{2\phi_1}}{\sqrt{e^{2\phi_1}+|p|^2}}-\frac{e^{2\phi_2}}{\sqrt{e^{2\phi_2}+|p|^2}}\right|\leq C_T|\phi_1-\phi_2|.
\]
Therefore, as $f_1\in L^\infty((0,T);L^1(\R^3))$,
\[
I_3\leq C_T\int_0^t\|\phi_1-\phi_2\|_{L^\infty((0,s))}\,ds.
\]
For $I_4$ we use H\"older's inequality, so that
\begin{align*}
\int_{\R^3}\frac{(f_1-f_2)}{\sqrt{e^{2\phi_2}+|p|^2}}\,dp &= \int_{\R^3} (e^{2\phi_2}+|p|^2)^{-\frac{1+\delta}{2}} \left ( (f_1-f_2) (e^{2\phi_2}+|p|^2)^{\frac{\delta}{2}} \right )  \,dp\\
& \leq \left (\int_{\R^3} (e^{2\phi_2}+|p|^2)^{-(1+\delta)} \, dp \right )^{1/2} \left (\int_{\R^3}  |f_1-f_2|^2 (e^{2\phi_2}+|p|^2)^\delta \,dp \right )^{1/2}\\
& \leq C\|(e^{2\phi_2}+|p|^2)^\frac{\delta}{2} (f_1-f_2)(t)\|_{L^2(\R^3)}\\
& \leq C_T\|(1+|p|^2)^\frac{\delta}{2} (f_1-f_2)(t)\|_{L^2(\R^3)}.
\end{align*}
Hence, collecting the estimates on $I_3$ and $I_4$ we obtain
\begin{equation}\label{deltaphi}
|(\phi_1-\phi_2)(t)|\leq C_T\int_0^t\Big(\|\phi_1-\phi_2\|_{L^\infty((0,s))}+\sup_{\tau \in (0,s)}\|(1+|p|^2)^\frac{\delta}{2} (f_1-f_2)(\tau)\|_{L^2(\R^3)}\Big)ds.
\end{equation}
Using~\eqref{deltaf} within~\eqref{deltaphi} we have, for all $t \in [0,T)$, 
\[
\|\phi_1-\phi_2\|_{L^\infty((0,t))}\leq C_T\int_0^t \|\phi_1-\phi_2\|_{L^\infty((0,s))} \,ds
\]
and conclude that $\phi_1=\phi_2$ and $f_1=f_2$ a.e. on $[0,T]\times\R^3$, for all $T>0$.

\section{Long time limit of the particle density in the ultra-relativistic case}
\label{longtime}
Our final purpose is to derive an explicit formula for the long time limit of solutions to~\eqref{ultrarelFPhom}.  
The results in this section require $\phi$ to satisfy $\phi\to -\infty$, as $t\to -\infty$, and
\[
\int_0^\infty e^{2\phi(t)}\,dt<\infty,
\]
which holds of course for the solutions of the VNFP system considered in the previous section.
To begin with we consider the ultra-relativistic Fokker-Planck equation which arises by setting $\phi \equiv -\infty$, or $e^{2\phi} \equiv 0$, within the diffusion matrix $D[\phi]$, namely
\begin{equation}\label{ultra}
\partial_t g=\partial_{p^i}\left(D^{ij}[-\infty]\partial_{p^j}g\right),\quad t>0,\ p\in\R^3,
\end{equation}
where $$D^{ij}[-\infty]=\frac{p^ip^j}{|p|}.$$
This is motivated by the fact that $e^{2\phi(t)} \to 0$ as $t \to \infty$, and hence one expects the asymptotic behavior of the density $f$ to mimic that of a solution to the reduced equation.

\begin{proposition}
\label{P1}
Let $g_\mathrm{in}\in L^1(\R^3)$ be given with $g_\mathrm{in}(p) \geq 0$ and let $g_\mathrm{in}(p)=\mathfrak{g}_\mathrm{in}(q,\omega)$ be the representation of $g_\mathrm{in}$ in spherical coordinates, where $q=|p|>0$, $\omega=p/q\in S^2$. There exists a unique global solution $g\in L^\infty((0,\infty);L^1(\R^3))$ of~\eqref{ultra} such that $g(0,p)=g_\mathrm{in}(p)$, which is given by $g(t,p)=\mathfrak{g}(t,q,\omega) \geq 0$, where 
\begin{equation}\label{ultrasol}
\mathfrak{g}(t,q,\omega)=\frac{e^{-\frac{q}{t}}}{tq}\int_0^\infty \mathfrak{g}_\mathrm{in}(z,\omega) z e^{-\frac{z}{t}}\,\mathcal{I}_2\left[2\frac{\sqrt{q}}{t}\sqrt{z}\right] dz,
\end{equation} 
and $\mathcal{I}_\alpha[x]$ denotes the $\alpha$th modified Bessel function of the first kind \cite{AS}. %
Moreover, if $g_\mathrm{in}\in L^\gamma(\R^3)$, then $g(t,\cdot)\in L^\gamma(\R^3)$ and we have the estimate
\begin{equation}\label{lpultra}
\|g(t)\|_{L^\gamma(\R^3)}\leq \|g_\mathrm{in}\|_{L^\gamma(\R^3)},
\end{equation}
for any $\gamma \in [1,\infty]$, where the equality holds for $\gamma=1$.
\end{proposition}

\begin{proof}
First we observe that the function $g(t,p)=\mathfrak{g}(t,q,\omega)$ given by~\eqref{ultrasol} belongs to $L^\infty((0,\infty); L^1(\R^3))$. In fact we have
\begin{align*}
\int_{\R^3}g(t,p)\,dp&=\int_{S^2}\int_0^\infty q^2\mathfrak{g}(t,q,\omega)\,dq\,d\omega\\
&=\frac{1}{t}\int_{S^2}\int_0^\infty e^{-\frac{z}{t}}z\,\mathfrak{g}_\mathrm{in}(z,\omega)\left(\int_0^\infty e^{-\frac{q}{t}}q\,\mathcal{I}_2\left[2\frac{\sqrt{z}}{t}\sqrt{q}\right]dq\right) dz\,d\omega.
\end{align*}
The integral within round brackets equals $e^{z/t}tz$, and thus
\begin{equation}\label{ultrasoll1}
\int_{\R^3}g(t,p)\,dp=\int_{S^2}\int_0^\infty z^2\mathfrak{g}_\mathrm{in}(z,\omega)\,dz\,d\omega,\quad \text{i.e., }\ \|g(t)\|_{L^1(\R^3)}=\|g_\mathrm{in}\|_{L^1(\R^3)}.
\end{equation}
Similarly one can prove that $g(t,p)$ given by~\eqref{ultrasol} satisfies the estimate~\eqref{lpultra}.
By a simple approximation argument, it suffices to show that~\eqref{ultrasol} is the unique solution of~\eqref{ultra} with smooth initial data. To this purpose we note that the operator 
\[
\mcL u = \partial_{p^i}\left(\frac{p^ip^j}{|p|} \partial_{p^j} u \right)
\]
is purely radial. In fact, the expression of $\mathcal{L}$ in spherical coordinates is given by
\[
\mathcal{L}=q\partial_q^2+3\partial_q.
\]
We may therefore treat the angular variables as constant parameters, i.e., we may fix $\omega\in S^2$, define $v_\mathrm{in}^\omega(q)=\mathfrak{g}_\mathrm{in}(q,\omega)$ and look for the solution $v^\omega(t,q)$ of
\begin{equation}\label{ultra2}
\partial_t v=q\partial_q^2v+3\partial_qv,\quad v^\omega(0,q)=v_\mathrm{in}^\omega(q).
\end{equation}
The solution $g(t,p)=\mathfrak{g}(t,q,\omega)$ of our original problem is then given by $\mathfrak{g}(t,q,\omega)=v^\omega(t,q)$. Our next observation is that the function $\mathfrak{u}(t,r)$ defined by
\[
\mathfrak{u}(t,r)=v^\omega(t,\frac{r^2}{4})
\]
solves the spherically symmetric heat equation in six dimensions, namely
\begin{equation}\label{heatsphsym}
\partial_t \mathfrak{u}=\partial_r^2 \mathfrak{u}+\frac{5}{r}\partial_r\mathfrak{u}.
\end{equation}
Hence, we have reduced the problem to finding the solution of~\eqref{heatsphsym}.
Let $u(t,x)$ be the solution of the Cauchy problem
\begin{subequations}\label{cauchyHEAT}
\begin{align}
&\partial_tu =\Delta u,\quad t>0,\ x\in\R^6,\label{d-heat}\\
& u(0,x)=u_\mathrm{in}(x),\quad x\in\R^6,
\end{align}
\end{subequations}
that is
\begin{equation}\label{solheat}
u(t,x)=\frac{1}{(4\pi t)^{3}}\int_{\R^6}e^{-\frac{|x-y|^2}{4t}}u_\mathrm{in}(y)\,dy.
\end{equation}
Assuming spherical symmetry, i.e., $u(t,x)=\mathfrak{u}(t,r)$,  $u(0,x)=\mathfrak{u}_\mathrm{in}(r)$, with $r=|x|$, and passing to hyperspherical coordinates in the integral on the right  side of~\eqref{solheat} we obtain
\begin{align}
\mathfrak{u}(t,r)&=\frac{e^{-\frac{r^2}{4t}}}{(4\pi t)^{3}}\int_{\R^6}e^{-\frac{(|y|^2-2y\cdot x)}{4t}}\mathfrak{u}_\mathrm{in}(|y|)\,dy\nonumber\\
&=\frac{8\pi^2}{3}\frac{e^{-\frac{r^2}{4t}}}{(4\pi t)^{3}}\int_0^\infty \mathfrak{u}_\mathrm{in}(s)\,e^{-\frac{s^2}{4t}}s^5\int_0^\pi\exp\left(\frac{r s\cos\theta}{2 t}\right)\sin^4\theta\, d\theta\,ds.\label{tempou}
\end{align}
Evaluating the angular integral gives
\[
\int_0^\pi\exp\left(\frac{r s\cos\theta}{2 t}\right)\sin^4\theta\,d\theta=12\pi\left(\frac{t}{rs}\right)^2\mathcal{I}_2\left[\frac{rs}{2t}\right],
\]
where $\mathcal{I}_\alpha[x]$ denotes the $\alpha$th modified Bessel function of the first kind \cite{AS}. Substituting this expression into~\eqref{tempou} we obtain
\begin{equation}\label{solheateqsphsym}
\mathfrak{u}(t,r)=\frac{1}{2}\frac{e^{-\frac{r^2}{4t}}}{r^2t}\int_0^\infty
\mathfrak{u}_\mathrm{in}(s)e^{-\frac{s^2}{4t}}s^3 \,\mathcal{I}_2\left[\frac{rs}{2t}\right]ds.
\end{equation}
Hence the solution of~\eqref{ultra2} is given by
\begin{equation}\label{solultra}
v^\omega(t,q)=\frac{e^{-\frac{q}{t}}}{tq}\int_0^\infty v^\omega_\mathrm{in}(z)z e^{-\frac{z}{t}}\,\mathcal{I}_2\left[2\frac{\sqrt{q}}{t}\sqrt{z}\right]\,dz,
\end{equation}
which is~\eqref{ultrasol}. 

\end{proof}

\begin{figure}
\begin{center}
\includegraphics[width=0.6\textwidth]{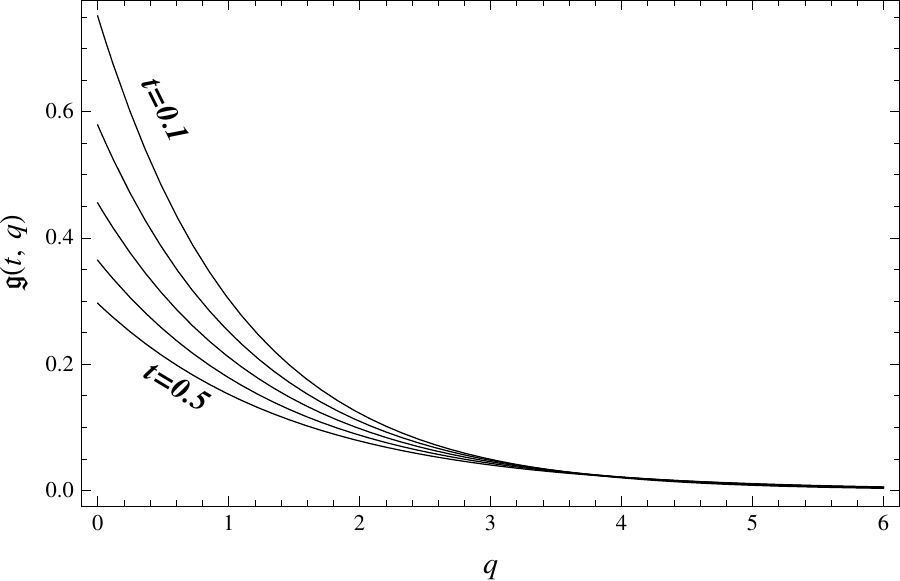}	
\caption{Numerical depiction of a spherically symmetric solution $\mathfrak{g}(t,q)$ of the ultra-relativistic Fokker-Planck equation~\eqref{ultra}. The picture shows the sections $t=const.$ of $\mathfrak{g}$ from $t=0.1$ (top curve) until $t=0.5$ (bottom curve). The initial datum is taken to be $\mathfrak{g}_\mathrm{in}(q)=e^{-q}$.}
\label{figultra}
\end{center}
\end{figure}

We remark, in particular, that when the initial datum is spherically symmetric, i.e., $g_\mathrm{in}(p)=\mathfrak{g}_\mathrm{in}(q)$, the solution of~\eqref{ultra} is also spherically symmetric, i.e., $g(t,p)=\mathfrak{g}(t,q)$, where
\begin{equation}\label{ultrasolss}
\mathfrak{g}(t,q)=\frac{e^{-\frac{q}{t}}}{tq}\int_0^\infty \mathfrak{g}_\mathrm{in}(z) z e^{-\frac{z}{t}}\,\mathcal{I}_2\left[2\frac{\sqrt{q}}{t}\sqrt{z}\right] dz.
\end{equation}
Figure~\ref{figultra} contains a numerical depiction of~\eqref{ultrasolss} for a few specific times and with the initial datum $\mathfrak{g}_\mathrm{in}(q)=e^{-q}$. 

With the preceding result, we can answer the analogous question for the ultra-relativistic equation with a scalar field using a simple change of variables.  

\begin{corollary}
\label{C1}
The solution of 
\begin{equation}\label{cappo}
\partial_t h= e^{2\phi} \partial_{p^i}\left(\frac{p^ip^j}{|p|}\partial_{p^j}h\right),
\end{equation}
with initial datum $h(0,p)=h_\mathrm{in}(p) \geq 0$ and $h_\mathrm{in} \in L^1(\mathbb{R}^3)$ is given by
\begin{equation}\label{rescaled}
h(t,p)=g(\tau(t),p),
\end{equation}
where 
\[
\tau(t)=\int_0^t e^{2\phi(s)}\,ds,
\]
and where $g(t,p)$ is the solution of~\eqref{ultra} with the same initial datum.

\end{corollary}

\begin{proof}

The result follows by rescaling time to account for the gravitational potential.
Let $$\tau(t) = \int_0^t e^{2\phi(s)} \ ds$$ and make the change of variables $h(t, p) = g(\tau(t), p)$.  Then, $\partial_t h = e^{2\phi} \partial_\tau g$, and the unknown function $g(\tau, p)$ satisfies the parabolic equation
$$ \partial_\tau g = \partial_{p^i} \left ( \frac{p^i p^j}{\vert p\vert} \partial_{p^j} g \right ).$$  
Hence any solution must be of the form (\ref{rescaled}).  

\end{proof}

Before proving our final result it is convenient to introduce some notation. Given $\phi=\phi(t)$ such that $e^{2\phi}\in L^1((0,\infty))$, and a function $u=u(p)$ with the representation $u(p)=\mathfrak{u}(q,\omega)$ in spherical coordinates, we define
\begin{subequations}\label{Tdef}
\begin{equation}
\mathcal{T}_\phi[u]=\frac{e^{-\frac{q}{\tau_\infty}}}{\tau_\infty q}\int_0^\infty \mathfrak{u}(z,\omega) z e^{-\frac{z}{\tau_\infty}}\,\mathcal{I}_2\left[2\frac{\sqrt{q}}{\tau_\infty}\sqrt{z}\right] dz,
\end{equation}
where 
\begin{equation}
\tau_\infty=\|e^{2\phi}\|_{L^1(\R^3)}.
\end{equation}
\end{subequations}
We are now in the position to derive the long time limit of solutions $h$ to~\eqref{cappo}. To avoid the need for technical estimates on Bessel functions in Lebesgue spaces, we choose to study the limit in the $L^\infty$ norm.
\begin{theorem}\label{hasym}
Let $\phi=\phi(t)$ be such that $e^{2\phi}\in L^1((0,\infty))$ and $h(t,p)\geq 0$ be the solution of 
\[
\partial_t h= e^{2\phi} \partial_{p^i}\left(\frac{p^ip^j}{|p|}\partial_{p^j}h\right)
\]
with initial datum $h(0,p)=h_\mathrm{in}(p)=\mathfrak{h}_\mathrm{in}(q,\omega)\geq 0$, cf. Corollary~\ref{C1}. We assume that the spherically symmetric function $\bar{\mathfrak{h}}_\mathrm{in}(q)=\sup_{\omega\in S^2}\mathfrak{h}_\mathrm{in}(q,\omega)$ satisfies
\begin{equation}\label{teches}
\int_0^\infty \bar{\mathfrak{h}}_\mathrm{in}(q) q\, (1+q)^2\,dq<\infty.
\end{equation}
Then, for all $t>1$,
\begin{equation}\label{asymh}
\|h(t)-\mathcal{T}_\phi[h_\mathrm{in}]\|_{L^\infty(\R^3)}\leq C \int_t^\infty e^{2\phi(s)}\,ds,
\end{equation}
where $C$ depends on $\tau(1)$ and $\tau_\infty$.
\end{theorem}

\begin{proof}
We have
\begin{align*}
h-\mathcal{T}_\phi[h_\mathrm{in}]&=\frac{e^{-\frac{q}{\tau(t)}}}{\tau(t)q}\int_0^\infty \mathfrak{h}_\mathrm{in}(z,\omega)ze^{-\frac{z}{\tau(t)}}\mathcal{I}_2\left[2\frac{\sqrt{q}}{\tau(t)}\sqrt{z}\right]dz\\
&\quad-\frac{e^{-\frac{q}{\tau_\infty}}}{\tau_\infty q}\int_0^\infty \mathfrak{h}_\mathrm{in}(z,\omega)ze^{-\frac{z}{\tau_\infty}}\mathcal{I}_2\left[2\frac{\sqrt{q}}{\tau_\infty}\sqrt{z}\right]dz\\
&=\frac{1}{q}\int_0^\infty \mathfrak{h}_\mathrm{in}(z,\omega)(H(\tau(t),q,z)-H(\tau_\infty,q,z)) z\,dz,
\end{align*}
where
\[
H(\tau,q,z):=\tau^{-1}e^{-\frac{q+z}{\tau}}\mathcal{I}_2\left[2\frac{\sqrt{q}}{\tau}\sqrt{z}\right].
\]
Since $t>1$, and $\tau(t)$ is increasing, we have $\tau_1:=\tau(1)<\tau(t)<\tau_\infty$.  By the Mean Value Theorem we estimate
\[
|h-\mathcal{T}_\phi[h_\mathrm{in}]|\leq \left(\int_t^\infty e^{2\phi(s)}ds\right)\frac{1}{q}\int_0^\infty  \mathfrak{h}_\mathrm{in}(z,\omega) \sup_{\tau\in(\tau_1,\tau_\infty)}|\partial_\tau H(\tau,q,z)| z\,dz.
\]
We prove below that
\begin{equation}\label{claimH}
\frac{1}{q}\sup_{\tau\in(\tau_1,\tau_\infty)}|\partial_\tau H(\tau,q,z)|\leq C(1+z)^2.
\end{equation}
Hence, having assumed \eqref{teches} we obtain
\[
|h-\mathcal{T}_\phi[h_\mathrm{in}]|\leq 
C\left(\int_t^\infty e^{2\phi(s)}ds\right)\int_0^\infty \bar{\mathfrak{h}}_\mathrm{in}(z) z(1+ z)^2\,dz \leq
C \int_t^\infty e^{2\phi(s)}ds,
\] 
which proves~\eqref{asymh}. It remains to establish~\eqref{claimH}.  Using the recurrence relation $$\mathcal{I}_2'[x]=\mathcal{I}_1[x]-\frac{2}{x}\mathcal{I}_2[x],$$ we find
\[
\partial_\tau H=\frac{e^{-\frac{q+z}{\tau}}}{\tau^2}\left(\mathcal{I}_2\left[2\frac{\sqrt{q}}{\tau}\sqrt{z}\right] \left (1+\frac{q+z}{\tau} \right)-2\frac{\sqrt{q}}{\tau}\sqrt{z}\,\mathcal{I}_1\left[2\frac{\sqrt{q}}{\tau}\sqrt{z}\right]\right).
\]
We shall use the following bounds satisfied by the modified Bessel functions~\cite{AS, Kreh}: for any $\alpha \in \mathbb{N}$, there exists $C > 0$ such that
$$ \mathcal{I}_\alpha[x]  \leq Cx^\alpha$$ for $x \leq 1$, while for $x \geq 1$ we have
$$ \mathcal{I}_\alpha[x] \leq Ce^x.$$ 
Thus, for $2\sqrt{qz}/\tau\leq 1$ we obtain
\[
|\partial_\tau H|\leq C\frac{e^{-\frac{q+z}{\tau}}}{\tau^4}qz \left (1+\frac{q+z}{\tau} \right),
\]
whence
\[
\frac{1}{q}\sup_{\tau\in(\tau_1,\tau_\infty)}|\partial_\tau H(\tau,q,z)|\leq C z(1+q+z)\leq C(1+z)^2.
\]
For $2\sqrt{qz}/\tau\geq 1$ we have 
\[
\frac{1}{q}|\partial_\tau H|\leq \frac{C}{\tau^2}e^{-\frac{(\sqrt{q}-\sqrt{z})^2}{\tau}}\left(\frac{1}{q}+\frac{1+z/q}{\tau}+\frac{\sqrt{z}}{\tau\sqrt{q}}\right)\leq C(1+z)^2,
\]
for $\tau\in (\tau_1,\tau_\infty)$. This completes the proof of the theorem.
\end{proof}
We remark that the technical condition~\eqref{teches} is just a little stronger than requiring that the first moment of $h_\mathrm{in}(p)$ is bounded in $L^1(\R^3)$.  In particular, if $h_\mathrm{in}$ is spherically symmetric then~\eqref{teches} is implied by any initial data $h_\mathrm{in} \in X$ satisfying \eqref{secondmombounded}.

\appendix
\setcounter{secnumdepth}{0}
\section{Appendix: Properties of the diffusion matrix}

In this appendix we collect some properties of the diffusion matrix
\[
D^{ij}[\phi]=\frac{e^{2\phi}\delta^{i j}+p^{i}p^{j}}{\sqrt{e^{2\phi}+|p|^2}},
\]
and other quantities which are used in the main body of the paper. We begin by proving~\eqref{boundsderD}.  Clearly 
\[
|D^{ij}[\phi]|\leq \sqrt{e^{2\phi}+|p|^2}.
\]
The first derivatives of $D$ are given by
\[
\partial_{p^k}D^{ij}=\frac{\delta^i_{\ k}p^j+\delta^j_{\ k}p^i}{\sqrt{e^{2\phi}+|p|^2}}-\frac{D^{ij}p_k}{e^{2\phi}+|p|^2},
\]
and therefore
\[
|\partial_{p^k}D^{ij}|\leq C.
\]
Moreover, we see that 
\[
\partial_{p^k}\partial_{p^l}D^{ij}=\frac{\delta^i_{\ k}\delta^j_{\ l}+\delta^j_{\ k}\delta^i_{\ l}}{\sqrt{e^{2\phi}+|p|^2}}-\frac{(\delta^i_{\ k}p^j+\delta^j_{\ k}p^i) p_l}{(e^{2\phi}+|p|^2)^{3/2}}-\frac{\partial_{p^l}D^{ij}p_k}{e^{2\phi}+|p|^2}-\frac{D^{ij}\delta_{kl}}{e^{2\phi}+|p|^2}+2\frac{D^{ij}p_kp_l}{(e^{2\phi}+|p|^2)^2}
\]
and each term in the right hand side is bounded in modulus by $C e^{-\phi}$, which proves the first estimate in~\eqref{boundsderD}. Furthermore
\[
\partial_{p^l}(\partial_{p^k}\partial_{p^i}D^{ij})=-3\frac{\delta^j_{\ k}p_l}{(e^{2\phi}+|p|^2)^{3/2}}-3\frac{\delta^j_{\ l}p_k+\delta_{kl}p^j}{(e^{2\phi}+|p|^2)^{3/2}}+9\frac{p^jp_kp_l}{(e^{2\phi}+|p|^3)^{5/2}},
\]
and each term in the right hand side is bounded in modulus by $C e^{-2\phi}$, which proves the second estimate in~\eqref{boundsderD}.

%
%

In order to justify the computation (\ref{Ddiff}), we first consider the function $$L(\psi) = e^{2\psi} D^{ij}[\psi]$$
and compute
$$L'(\psi) = \frac{e^{2\psi}}{\sqrt{e^{2\psi} + \vert p \vert^2}} \left [3e^{2\psi} \delta^{ij} + p^i p^j \right ] 
+  \frac{\vert p \vert^2 e^{2\psi}}{(e^{2\psi} + \vert p \vert^2)^{3/2}} \left [ e^{2\psi} \delta^{ij} + p^i p^j \right ].$$
Thus, for all bounded $\psi$,
\begin{align*}
\vert L'(\psi) \vert & \leq  e^{2\psi} e^{-\psi} \cdot 3e^{2\psi} + e^{2\psi} \vert p \vert^{-1} \left \vert p^i p^j \right \vert +  \vert p \vert^2 e^{2\psi} \left ( e^\psi \vert p \vert \right)^{-3/2} e^{2\psi} + 
\vert p \vert^2 e^{2\psi} \vert p\vert^{-3} \left \vert p^i p^j \right \vert \\
& \leq  4 e^{3\psi} + e^{2\psi}\vert p \vert^{1/2}e^{\psi/2} + 2e^{2\psi} \vert p \vert \\ 
& \leq C\sqrt{1 + \vert p \vert^2}.
\end{align*}
The inequality (\ref{Ddiffderiv}) follows similarly.

Additionally, we must verify the smoothness of coefficients within the system of stochastic differential equations (\ref{sde}) in order to arrive at the existence of a unique solution, where the vector $d$ and the matrix $G$ are given in~\eqref{ddef}-\eqref{Gdef}.
We first compute 
$$ \partial_{p^j} d^i(t,p) = \frac{3e^{2\bar{\phi}}}{(e^{2\bar{\phi}} + \vert p \vert^2)^{-3/2}} \left ( \delta^{ij}e^{2\bar{\phi}} + \delta^{ij} \vert p \vert^2 - p^i p^j \right ).$$
So, estimating we find for every $i,j$
\begin{align*}
\vert  \partial_{p^j} d^i(t,p) \vert & \leq  Ce^{2\bar{\phi}} (e^{2\bar{\phi}} + \vert p \vert^2 )^{-3/2} \left [e^{2\bar{\phi}} + \vert p \vert^2 \right ]\\
& \leq   Ce^{2\bar{\phi}} (e^{2\bar{\phi}} + \vert p \vert^2 )^{-1/2}\\
& \leq  C_T
\end{align*}
and derivatives are uniformly bounded.  Second derivatives can be computed and bounded similarly.

Next, we find
\begin{align*}
\partial_{p_k} G^{ij}(t,p) & =   \sqrt{2} e^{\bar{\phi}}\left (e^{2\bar{\phi}} + \vert p \vert^2 \right)^{-5/4} \left (e^{\bar{\phi}} + \sqrt{e^{2\bar{\phi}} + \vert p \vert^2} \right)^{-2} \left [ -\frac{1}{2} p_k p^i p^j e^{\bar{\phi}} - \frac{3}{2} p_k p^i p^j \sqrt{e^{2\bar{\phi}} + \vert p \vert^2} \right.\\
&\quad \left. + \ (e^{2\bar{\phi}} + \vert p \vert^2) \left (e^{\bar{\phi}} 
+ \sqrt{e^{2\bar{\phi}} + \vert p \vert^2} \right ) (\delta^i_kp^j + \delta^j_k p^i) - \frac{1}{2} e^{\bar{\phi}} \left (e^{\bar{\phi}}+\sqrt{e^{2\bar{\phi}} + \vert p \vert^2} \right)^2 p_k \delta^{ij}
 \right ]
 \end{align*}
and thus for every $i,j,k$
\begin{align*}
\vert \partial_{p_k} G^{ij}(t,p)\vert & \leq  C e^{\bar{\phi}}\left (e^{2\bar{\phi}} + \vert p \vert^2 \right)^{-5/4} \left (e^{\bar{\phi}} + \sqrt{e^{2\bar{\phi}} + \vert p \vert^2} \right)^{-2}  \\ 
&\quad\times \left [ \vert p \vert^3 \left (e^{\bar{\phi}} + \sqrt{e^{2\bar{\phi}} + \vert p \vert^2} \right) + \vert p \vert \left (e^{2\bar{\phi}} + \vert p \vert^2 \right) \left (e^{\bar{\phi}} + \sqrt{e^{2\bar{\phi}} + \vert p \vert^2} \right)\right.\\
&\qquad\qquad\qquad\left. + e^{\bar{\phi}} \vert p \vert \left (e^{\bar{\phi}} + \sqrt{e^{2\bar{\phi}} + \vert p \vert^2} \right)^2 \right ] \\
& \leq  C\vert p \vert e^{\bar{\phi}}\left (e^{2\bar{\phi}} + \vert p \vert^2 \right)^{-5/4} \left (e^{\bar{\phi}} + \sqrt{e^{2\bar{\phi}} + \vert p \vert^2} \right)^{-2} \cdot \\
&\quad \left [ \left (e^{2\bar{\phi}} + \vert p \vert^2 \right) \left (e^{\bar{\phi}} + \sqrt{e^{2\bar{\phi}} + \vert p \vert^2} \right)
+ e^{\bar{\phi}} \left (e^{\bar{\phi}} + \sqrt{e^{2\bar{\phi}} + \vert p \vert^2} \right)^2  \right ]   \\
& \leq   C \vert p \vert e^{\bar{\phi}}\left (e^{2\bar{\phi}} + \vert p \vert^2 \right)^{-1/4} \left (e^{\bar{\phi}} + \sqrt{e^{2\bar{\phi}} + \vert p \vert^2} \right)^{-1}
+ C \vert p \vert e^{2\bar{\phi}}\left (e^{2\bar{\phi}} + \vert p \vert^2 \right)^{-5/4} \\
& =  I + II.
\end{align*}
The first term satisfies $$ I \leq C e^{\bar{\phi}}\left (e^{2\bar{\phi}} + \vert p \vert^2 \right)^{-1/4}  \leq C e^{\bar{\phi}/2}.$$
Alternatively, the second term can be estimated for $\vert p \vert \leq e^{\bar{\phi}}$ as
$$II \leq C e^{\bar{\phi}} e^{2\bar{\phi}} e^{-5\bar{\phi}/2}  \leq Ce^{\bar{\phi}/2}$$
and for $\vert p \vert \geq e^{\bar{\phi}}$ as
$$II \leq C \vert p \vert e^{2\bar{\phi}} \left (\vert p \vert e^{\bar{\phi}} \right )^{-5/4}  \leq C\vert p \vert^{-1/4} e^{3\bar{\phi}/4} \leq Ce^{\bar{\phi}/2}. $$
Combining the estimates we conclude
$$\vert \partial_{p_k} G^{ij}(t,p)\vert \leq Ce^{\bar{\phi}/2} \leq C_T$$
for all $i, j, k$. Hence, derivatives are uniformly bounded in the time interval $[-T,0]$. As before, second derivatives can be bounded using similar estimates.

{\bf Acknowledgments:} The first author is sponsored by the Mexican National Council for Science and Technology (CONACYT) with scholarship number 214152. The last author is supported by the US National Science Foundation under grants DMS-0908413 and DMS-1211667.  In addition, the authors wish to thank the anonymous reviewers for careful analysis of the original manuscript and constructive feedback that served to improve the content of the paper.

\end{document}